\documentclass[letterpaper,10pt,conference]{ieeeconf}
\IEEEoverridecommandlockouts
\makeatletter
\@ifundefined{keywords}{}{}
\makeatother

\usepackage{amsmath,amssymb,amsfonts}
\usepackage{graphicx}
\usepackage{bm}
\usepackage{mathtools}

\makeatletter
\@ifundefined{proof}{\usepackage{amsthm}}{}
\@ifundefined{theoremstyle}{\newcommand{\theoremstyle}[1]{}}{}
\makeatother
\theoremstyle{definition}
\newtheorem{assumption}{Assumption}
\newtheorem{property}{Property}
\newtheorem{theorem}{Theorem}
\newtheorem{lemma}{Lemma}
\newtheorem{remark}{Remark}
\begin{document}

\title{\LARGE \bf
Reduced-Order Backstepping Control of Nonlinear Mechanical Systems with Input Delay via Galerkin Projection
}

\author{Amirhossein Asadi and Nitin Sharma\thanks{A. Asadi is with the NC State Department of Mechanical and Aerospace Engineering, NC State University, Raleigh, NC 27606 USA (e-mail: aasadi@ncsu.edu). N. Sharma is with the UNC/NC State Joint Department of Biomedical Engineering, NC State University, Raleigh, NC 27606 USA (e-mail: nsharm23@ncsu.edu). *Corresponding author: Nitin Sharma.}}

\maketitle
\thispagestyle{empty}
\pagestyle{empty}

\begin{abstract}
This paper addresses trajectory tracking for a class of nonlinear mechanical systems with input delay. The delayed actuator channel is represented as an equivalent transport partial differential equation, thereby transforming the original delayed system into a coupled plant–channel representation. A finite-dimensional approximation of the delay channel is then constructed through Galerkin projection, resulting in an explicit ODE model of the delayed actuator dynamics. Using this reduced-order plant–channel model, a backstepping controller is developed, and Lyapunov analysis establishes uniform exponential stability of the nominal tracking error dynamics for the reduced-order Galerkin model. To enhance practical implementation, a command-filtered version is also introduced to avoid direct differentiation of virtual controls, for which a corresponding uniform ultimate boundedness result is obtained. Simulation studies are provided to validate the proposed design and illustrate the effectiveness of the method.  
\end{abstract}

\section{Introduction}

Input delays are a common source of performance degradation in nonlinear mechanical systems, as they can weaken tracking accuracy and even compromise closed-loop stability. This is particularly important when fast and reliable actuation is required but delay is unavoidable. Some studies address input delay from a robustness perspective. For example, \cite{yue2004robust} developed an LMI-based robust stabilization method for uncertain systems with input delay, while \cite{mazenc2013robustness} analyzed the robustness of nonlinear feedback laws under delayed and sampled actuation using Lyapunov analysis. These approaches mainly treat delay as a perturbation to be bounded and tolerated rather than as a dynamic channel to be explicitly compensated.

Another major line of research addresses input delay through explicit compensation, most commonly using predictor-based methods. For nonlinear mechanical systems, \cite{sharma2011predictor} proposed predictor-like PID/PD controllers for Euler--Lagrange dynamics with known constant input delay and established semi-global uniformly ultimately bounded tracking via Lyapunov--Krasovskii analysis. In an application-specific setting, \cite{alibeji2017modified} incorporated delayed activation dynamics into neuromuscular electrical stimulation control using modified dynamic surface control with online state estimation, demonstrating the benefit of explicitly accounting for actuator-side delay. PDE-based viewpoints have also been developed for actuator-delay systems, including first-order hyperbolic delay-channel representations and control of uncertain Euler--Lagrange dynamics \cite{krstic2008backstepping, chakraborty2016control}. More recent PDE-based formulations represent actuator delay as a transport equation and compensate it through predictor or backstepping constructions. For example, \cite{zekraoui2023finite} developed an infinite-dimensional nonlinear backstepping framework for delayed chains of integrators, while \cite{lai2025indirect} combined PDE-based predictor feedback, batch parameter estimation, and robust compensation for input-delay systems with unknown parameters and periodic disturbances. Related developments include approximate neural-operator predictors for nonlinear actuator-delay systems \cite{bhan2025delay} and predictor-based delay adaptation with safe backstepping for strict-feedback systems with unknown delay \cite{zhao2025safe}.

A related direction focuses on reducing the implementation complexity of recursive nonlinear designs. Dynamic surface control provides a filtered alternative to standard backstepping: the foundational work \cite{swaroop2002dynamic} used low-pass filters to avoid repeated differentiation of virtual controls and mitigate the explosion-of-terms problem. In a delayed setting, \cite{shojaei2023observer} developed observer-based adaptive dynamic surface control for uncertain time-delayed non-strict nonlinear systems with input nonlinearities, unknown control direction, and output constraints using Lyapunov--Krasovskii and barrier Lyapunov functions, fuzzy approximation, and state observation. These works emphasize reducing controller complexity while accommodating nonlinear effects and delays within a Lyapunov-based framework.

Taken together, existing delay-control strategies mainly compensate delay through predictor-based constructions or accommodate it within robust or Lyapunov-based designs. In contrast, the present work adopts a reduced-order modeling viewpoint, treating the delayed actuator channel as a dynamic object to be approximated and incorporated directly into the controller synthesis. This perspective is motivated by Galerkin-based studies of delay systems, where low-dimensional approximations reduce the computational burden associated with the infinite-dimensional state and provide accurate ODE surrogates for analysis and control \cite{wahi2005galerkin}. In particular, \cite{chekroun2017galerkin} developed Galerkin-Koornwinder approximations for nonlinear delay differential equations for control-oriented computations, while \cite{uchida2016stability} and \cite{sharafi2023stability} used Galerkin projection to convert delayed inverted-pendulum models into finite-dimensional ODE systems for stability analysis. Building on this viewpoint, we represent the actuator delay as a transport PDE and project it onto a finite-dimensional plant--channel model for controller design. The main contributions of this work are as follows. First, we adopt a reduced-order modeling viewpoint for input-delay compensation in nonlinear mechanical systems by projecting the transport-PDE actuator channel onto a Galerkin-based finite-dimensional model. Based on this approximation, we develop a backstepping controller that explicitly incorporates the channel-state dynamics. We further establish rigorous closed-loop guarantees through Lyapunov analysis, including uniform exponential stability for the reduced-order Galerkin design and uniform ultimate boundedness for a command-filtered implementation. Finally, simulation studies demonstrate the effectiveness of the proposed method and support the theoretical results.

\section{Problem Formulation}

Consider the nonlinear mechanical system with the input delay
\begin{equation}
M(q)\ddot q + C(q,\dot q)\dot q + G(q)=u(t-\tau),
\label{eq:1}
\end{equation}
where $q(t)\in\mathbb{R}^n$, $M(q)\in\mathbb{R}^{n\times n}$ is the inertia matrix, $C(q,\dot q)\in\mathbb{R}^{n\times n}$ denotes the Coriolis and centripetal matrix, $G(q)\in\mathbb{R}^n$ is the gravity vector, and $u(t-\tau)\in\mathbb{R}^n$ is the delayed control input subjected to a constant time delay $\tau>0$. The control objective is to design a feedback law such that the generalized coordinate $q(t)$ tracks a prescribed reference trajectory $q_d(t)\in\mathbb{R}^n$ despite the presence of the actuator delay. The following standard assumptions and one structural property are imposed throughout the paper.

\begin{assumption}
The mappings $M(q)$, $C(q,\dot q)$, and $G(q)$ are continuously differentiable on a compact operating region $\Omega \subset \mathbb{R}^{2n}$, and their values and first derivatives are bounded on $\Omega$. Moreover, $M(q)$ is symmetric positive definite, and there exist constants $0<m\le \bar m$ such that $mI \le M(q) \le \bar m I$ for all $(q,\dot q)\in\Omega$.
\end{assumption}

\begin{assumption}
The desired trajectory $q_d(t) \in \mathcal{C}^3$, with 
$q_d(t)$, $\dot{q}_d(t)$, $\ddot{q}_d(t)$, and $q_d^{(3)}(t)$ 
uniformly bounded.
\end{assumption}

\begin{property}
The mechanical system satisfies the standard skew-symmetry property, namely $\dot M(q)-2C(q,\dot q)$ is skew-symmetric for all $(q,\dot q)$ in the domain of interest.
\end{property}

\section{Galerkin Approximation of the Input-Delay Channel}

The presence of the delayed input in \eqref{eq:1} makes the closed-loop system infinite-dimensional, since the current plant behavior depends not only on the instantaneous input but also on the input history over the interval $[t-\tau,t]$. To handle this delay, we represent the actuator channel by a transport PDE and then approximate that PDE by a finite-dimensional model. The purpose of this section is therefore to replace the delayed-input system by an approximated reduced-order plant--channel interconnection that is more suitable for analysis and controller design.

To this end, we introduce the dynamic extension
\begin{equation}
\dot u(t)=\nu(t),
\label{eq:2}
\end{equation}
where $\nu(t)\in\mathbb{R}^n$ is a new control variable to be designed. This dynamic extension allows the delayed input to be represented as the state of a transport system. Inspired by \cite{krstic2009input}, we define the distributed state
\begin{equation}
y(s,t)=u(t+s), 
\label{eq:3}
\end{equation}
where $y(s,t)\in\mathbb{R}^n$ and $s\in[-\tau,0]$. Here, the spatial variable $s$ parameterizes the input history, with $s=0$ corresponding to the current command and $s=-\tau$ corresponding to the delayed signal that actually reaches the plant. Differentiating \eqref{eq:3} with respect to $t$ and $s$ gives the transport equation
\begin{equation}
\frac{\partial y(s,t)}{\partial t}
=
\frac{\partial y(s,t)}{\partial s}.
\label{eq:4a}
\end{equation}
subject to the boundary condition
\begin{equation}
y(0,t)=u(t).
\label{eq:4b}
\end{equation}
The delayed signal applied to the mechanical system is recovered at the left boundary of the transport domain:
\begin{equation}
u(t-\tau)=y(-\tau,t).
\label{eq:5}
\end{equation}
Thus, the input delay is represented exactly by a first-order PDE defined on the finite interval $[-\tau,0]$. 

To derive a finite-dimensional approximation, let $\phi(s)=[\phi_1(s),\phi_2(s),\ldots,\phi_N(s)]^\top$, where $\phi_i(s)$, $i=1,\dots,N$, are linearly independent basis functions. In practice, shifted Legendre polynomials are a convenient choice because of their orthogonality and favorable approximation properties on bounded intervals.

Since the transport state $y(s,t)$ is vector-valued, define the block basis matrix $\Phi(s)=\phi(s)\otimes I_n\in\mathbb{R}^{nN\times n}$ and approximate the distributed state by
\begin{equation}
y_N(s,t)=\Phi^\top(s)a(t).
\label{eq:8}
\end{equation}
The vector $a(t)\in\mathbb{R}^{nN}$ contains the modal coefficients of the truncated expansion and serves as the finite-dimensional state of the approximated delay channel.

The next step is to derive an evolution equation for $a(t)$. To impose the boundary condition in a weak form, introduce a Lagrange multiplier $\lambda(t)\in\mathbb{R}^n$ and consider
\begin{equation}
\frac{\partial y_N(s,t)}{\partial t}
=
\frac{\partial y_N(s,t)}{\partial s}
+\delta(s)\lambda(t).
\label{eq:9}
\end{equation}
where $\delta(s)$ denotes the Dirac delta distribution. The additional term acts only at the boundary and enforces the input constraint within the Galerkin framework.

Substituting \eqref{eq:8} into \eqref{eq:9}, premultiplying by $\Phi(s)$, and integrating over $[-\tau,0]$ yield the Galerkin equations
\begin{equation}
\mathcal{M}_N \dot a(t)=\mathcal{K}_N a(t)+\Phi(0)\lambda(t),
\label{eq:10}
\end{equation}
where
\begin{equation}
\mathcal{M}_N=\int_{-\tau}^{0}\Phi\Phi^\top ds,
\qquad
\mathcal{K}_N=\int_{-\tau}^{0}\Phi\frac{d\Phi^\top}{ds}ds,
\label{eq:MK}
\end{equation}
The matrix $\mathcal{M}_N\in\mathbb{R}^{nN\times nN}$ is the Galerkin mass matrix, while $\mathcal{K}_N\in\mathbb{R}^{nN\times nN}$ is the projected representation of the spatial derivative operator appearing in the transport PDE.

To eliminate the auxiliary multiplier $\lambda(t)$, we next use the boundary dynamics. Differentiating the boundary condition and using
$\dot u=\nu$ gives $\partial_t y(0,t)=\nu(t)$; under
\eqref{eq:8}, this becomes
$\Phi^\top(0)\dot a(t)=\nu(t)$.
Premultiplying \eqref{eq:10} by $\Phi^\top(0)\mathcal{M}_N^{-1}$ and using this boundary relation give
\begin{equation}
\nu(t)=\mathcal{C}_{2,N}a(t)+\mathcal{C}_{3,N}\lambda(t),
\label{eq:14}
\end{equation}
where
\begin{equation}
\mathcal{C}_{2,N}
=\Phi^\top(0)\mathcal{M}_N^{-1}\mathcal{K}_N,
\qquad
\mathcal{C}_{3,N}
=\Phi^\top(0)\mathcal{M}_N^{-1}\Phi(0).
\label{eq:C23}
\end{equation}
By Lemma 1, $C_{3,N}$ is nonsingular, and the
Lagrange multiplier can therefore be solved explicitly as
\begin{equation}
\lambda(t)
=
\mathcal{C}_{3,N}^{-1}
\big(
\nu(t)-\mathcal{C}_{2,N}a(t)
\big).
\label{eq:16}
\end{equation}
Substituting \eqref{eq:16} back into \eqref{eq:10} removes the boundary multiplier and produces a standard finite-dimensional state-space model for the channel:
\begin{equation}
\dot a(t)=A_N a(t)+B_N \nu(t),
\label{eq:17}
\end{equation}
where
\begin{equation}
A_N
=
\mathcal{M}_N^{-1}
\Big(
\mathcal{K}_N-\Phi(0)\mathcal{C}_{3,N}^{-1}\mathcal{C}_{2,N}
\Big),
\label{eq:18a}
\end{equation}
\begin{equation}
B_N
=
\mathcal{M}_N^{-1}\Phi(0)\mathcal{C}_{3,N}^{-1}.
\label{eq:18b}
\end{equation}

Once the modal dynamics have been obtained, the delayed signal applied to the plant is recovered from the left boundary of the approximated channel:
\begin{equation}
y_N(-\tau,t)=\Phi^\top(-\tau)a(t)=C_N a(t),
\label{eq:19}
\end{equation}
where $C_N=\Phi^\top(-\tau)\in\mathbb{R}^{n\times nN}$.

\begin{lemma}
For the shifted Legendre basis
$\phi_i(s)=P_{i-1}(1+2s/\tau)$, $i=1,\ldots,N$,
\begin{equation}
C_{3,N}=\frac{N^2}{\tau}I_n,
\qquad
E_N\triangleq C_NB_N=\frac{(-1)^{N-1}}{N}I_n.
\end{equation}
Hence, $C_{3,N}$ and $E_N$ are nonsingular for every
$N\geq1$.

\end{lemma}

\begin{proof}
The result follows directly from Legendre orthogonality,
$P_i(1)=1$, and $P_i(-1)=(-1)^i$; details are omitted
for brevity.
\end{proof}

Accordingly, the original input-delay system \eqref{eq:1} is approximated by the finite-dimensional plant--channel interconnection
\begin{equation}
M(q)\ddot q + C(q,\dot q)\dot q + G(q)=C_N a,
\label{eq:21a}
\end{equation}
\begin{equation}
\dot a=A_N a+B_N \nu.
\label{eq:21b}
\end{equation}
This reduced-order model has channel-state dimension $nN$. In general, increasing the Galerkin order $N$ improves the accuracy of the delay approximation, although at the expense of a higher-dimensional controller realization.

\section{Control Design for the Reduced Order System}

In this section, it is assumed that the full state $(q,\dot q,a)$ is available for feedback. The control design proceeds by backstepping on the reduced-order plant--channel model \eqref{eq:21a}--\eqref{eq:21b}. The main difficulty is that the plant input $C_N a$ cannot be assigned directly; instead, it must be generated indirectly through the channel dynamics driven by the command-rate input $\nu$.

Introduce the tracking error variables
\begin{equation}
z_1(t)=q(t)-q_d(t),
\qquad
z_2(t)=\dot q(t)-\alpha_1(t),
\label{eq:25}
\end{equation}
where the first virtual control is chosen as
\begin{equation}
\alpha_1(t)=-K_1 z_1(t)+\dot q_d(t),
\label{eq:26}
\end{equation}
where $K_1\in\mathbb{R}^{n\times n}$ is a symmetric positive definite gain matrix. With this definition,
\begin{equation}
\dot z_1=z_2-K_1 z_1,
\label{eq:27}
\end{equation}
and
\begin{equation}
\dot z_2
=
M^{-1}(q)\!\left[C_N a-C(q,\dot q)\dot q-G(q)\right]-\dot\alpha_1.
\label{eq:28}
\end{equation}

Consider the Lyapunov function
\begin{equation}
V_1=\frac{1}{2}z_1^\top z_1
+\frac{1}{2}z_2^\top M(q)z_2.
\label{eq:V1}
\end{equation}
Using Property~1, its derivative is
\begin{equation}
\begin{aligned}
\dot V_1
= {} & z_1^\top z_2-z_1^\top K_1 z_1 \\
& +z_2^\top\!\left[ C_N a - C(q,\dot q)\alpha_1 - G(q) - M(q)\dot\alpha_1 \right].
\end{aligned}
\label{eq:32}
\end{equation}

At this point, the channel output $C_N a$ plays the role of a virtual delayed input that must be shaped to stabilize the mechanical subsystem. To this end, define the channel tracking error
\begin{equation}
z_3(t)=C_N a(t)-\alpha_2(t),
\label{eq:33}
\end{equation}
where $\alpha_2(t)\in\mathbb{R}^n$ is a second virtual control to be designed. Using \eqref{eq:33} to rewrite \eqref{eq:32} gives
\begin{equation}
\begin{aligned}
\dot V_1
= {} & z_1^\top z_2-z_1^\top K_1 z_1 \\
& +z_2^\top\!\left[ z_3+\alpha_2-C(q,\dot q)\alpha_1-G(q)-M(q)\dot\alpha_1 \right].
\end{aligned}
\label{eq:35}
\end{equation}

Choose the second virtual control as
\begin{equation}
\begin{aligned}
\alpha_2
= {} & -z_1-K_2 z_2+C(q,\dot q)\alpha_1+G(q)  +M(q)\dot\alpha_1.
\end{aligned}
\label{eq:36}
\end{equation}
where $K_2\in\mathbb{R}^{n\times n}$ is a symmetric positive definite gain matrix.
Under this choice, \eqref{eq:35} reduces to
\begin{equation}
\dot V_1
=
-z_1^\top K_1 z_1-z_2^\top K_2 z_2+z_2^\top z_3.
\label{eq:37}
\end{equation}
Thus, the remaining coupling term is entirely captured by $z_3$, which reflects the mismatch between the actual channel output and the desired virtual input.

It remains to design $\nu$ so that $C_N a$ tracks $\alpha_2$. Differentiating \eqref{eq:33} and using \eqref{eq:21b} yields
\begin{equation}
\dot z_3=C_NA_N a + C_NB_N\nu - \dot\alpha_2.
\label{eq:38}
\end{equation}
Here $E_N=C_NB_N$ is nonsingular by Lemma~1. The command-rate input is then selected as
\begin{equation}
\nu
=
E_N^{-1}
\left[
-C_NA_N a+\dot\alpha_2-z_2-K_3 z_3
\right].
\label{eq:40}
\end{equation}
where $K_3\in\mathbb{R}^{n\times n}$ is a symmetric positive definite gain matrix.
Substituting \eqref{eq:40} into \eqref{eq:38} gives the closed-loop channel error dynamics
\begin{equation}
\dot z_3=-z_2-K_3 z_3.
\label{eq:41}
\end{equation}

Now consider the composite Lyapunov function
\begin{equation}
V_2
=
\frac{1}{2}z_1^\top z_1
+\frac{1}{2}z_2^\top M(q)z_2
+\frac{1}{2}z_3^\top z_3.
\label{eq:42}
\end{equation}
Using \eqref{eq:37} and \eqref{eq:41}, its time derivative becomes
\begin{equation}
\begin{aligned}
\dot V_2
&=
-z_1^\top K_1 z_1-z_2^\top K_2 z_2+z_2^\top z_3
+z_3^\top(-z_2-K_3 z_3) \\
&=
-z_1^\top K_1 z_1-z_2^\top K_2 z_2-z_3^\top K_3 z_3.
\end{aligned}
\label{eq:43}
\end{equation}
This establishes that all error coordinates are dissipative under the proposed design.
\begin{theorem}
Consider the reduced-order system \eqref{eq:21a}--\eqref{eq:21b} under Assumptions 1--2 and Property 1. Suppose that the virtual controls and command-rate control are chosen according to \eqref{eq:26}, \eqref{eq:36}, and \eqref{eq:40}. Then, for all closed-loop trajectories evolving in the compact operating region $\Omega$, the origin of the corresponding closed-loop error system is uniformly exponentially stable. In particular, defining
\begin{equation}
\mathbf{z} \triangleq
\begin{bmatrix}
z_1^\top & z_2^\top & z_3^\top
\end{bmatrix}^\top,
\label{eq:43a}
\end{equation}
there exist positive constants $c_e$ and $\lambda_e$ such that
\begin{equation}
\|\mathbf{z}(t)\| \le c_e e^{-\lambda_e t}\|\mathbf{z}(0)\|, \qquad t\ge 0.
\label{eq:43b}
\end{equation}
Consequently, $z_1(t)$, $z_2(t)$, and $z_3(t)$ converge to zero exponentially as $t\to\infty$.
\end{theorem}

\begin{proof}
Define $\mathbf{z}$ as in \eqref{eq:43a}. By Assumption 1 and \eqref{eq:42}, there exist positive constants $c_1$ and $c_2$ such that
\begin{equation}
c_1=\frac{1}{2}\min\{1,m\}, \qquad
c_2=\frac{1}{2}\max\{1,\bar m\},
\label{eq:43c}
\end{equation}
and
\begin{equation}
c_1\|\mathbf{z}\|^2 \le V_2 \le c_2\|\mathbf{z}\|^2.
\label{eq:43d}
\end{equation}
Also, from \eqref{eq:43}, define
\begin{equation}
k \triangleq \min\left\{
\lambda_{\min}(K_1),\,
\lambda_{\min}(K_2),\,
\lambda_{\min}(K_3)
\right\}>0.
\label{eq:43e}
\end{equation}
Then
\begin{equation}
\dot V_2 \le -k\|\mathbf{z}\|^2 \le -\frac{k}{c_2}V_2.
\label{eq:43f}
\end{equation}
Integrating \eqref{eq:43f} and using \eqref{eq:43d} gives
\begin{equation}
\|\mathbf{z}(t)\|
\le
\sqrt{\frac{c_2}{c_1}}\,
e^{-\frac{k}{2c_2}t}\,
\|\mathbf{z}(0)\|,
\qquad t\ge 0.
\label{eq:43g}
\end{equation}
Hence, \eqref{eq:43b} holds with
$c_e=\sqrt{\frac{c_2}{c_1}}$ and $\lambda_e=\frac{k}{2c_2}$, which proves uniform exponential stability on $\Omega$.
\end{proof}

\section{Command-Filtered Control Design}

The command-rate law in \eqref{eq:40} requires the time derivative of the second virtual control. Since $\alpha_2$ already depends on nonlinear model terms together with the derivative of the first virtual control, direct computation of $\dot\alpha_2$ may be inconvenient in practice. To address this issue, a first-order command filter is introduced in place of the unfiltered virtual signal. The design follows a Lyapunov/backstepping framework, with the filtered implementation motivated by prior work \cite{farrell2009command} showing that filtering reduces differentiation burden while preserving backstepping stability.

Define the filtered virtual control $\alpha_{2,f}(t)\in\mathbb{R}^n$ by
\begin{equation}
\varepsilon_f \dot \alpha_{2,f}+\alpha_{2,f}=\alpha_2,
\qquad \varepsilon_f>0,
\label{eq:44}
\end{equation}
where $\varepsilon_f$ is a design constant. Next, define the filtered channel tracking error and the filter mismatch as
\begin{equation}
\tilde z_3=C_N a-\alpha_{2,f},
\qquad
\eta_2=\alpha_{2,f}-\alpha_2.
\label{eq:45}
\end{equation}
Then the channel output can be written as
\[
C_N a=\tilde z_3+\alpha_{2,f}
=\tilde z_3+\alpha_2+\eta_2.
\]
Substituting this relation into \eqref{eq:32} and using \eqref{eq:36} yields
\begin{equation}
\dot V_1
=
-z_1^\top K_1 z_1-z_2^\top K_2 z_2+z_2^\top \tilde z_3+z_2^\top \eta_2.
\label{eq:46}
\end{equation}

Differentiating the first relation in \eqref{eq:45} and using \eqref{eq:21b} give
\begin{equation}
\dot{\tilde z}_3=C_NA_N a + C_NB_N\nu-\dot\alpha_{2,f}.
\label{eq:47}
\end{equation}
Accordingly, the command-rate input is chosen as
\begin{equation}
\nu
=
E_N^{-1}
\left[
-C_NA_N a+\dot\alpha_{2,f}-z_2-K_3\tilde z_3
\right].
\label{eq:48}
\end{equation}
From \eqref{eq:44},
$\dot\alpha_{2,f}=-(1/\varepsilon_f)\eta_2$; hence,
\eqref{eq:48} requires no explicit evaluation of $\dot\alpha_2$.
Substituting \eqref{eq:48} into \eqref{eq:47} yields
$\dot{\tilde z}_3=-z_2-K_3\tilde z_3$.

\begin{assumption}
There exists $\bar d_2>0$ such that
$\|\dot\alpha_2(t)\|\le\bar d_2$ for all $t\ge0$
along closed-loop trajectories in the compact operating region.
\end{assumption}

Using \eqref{eq:44} and \eqref{eq:45}, the filter mismatch dynamics satisfy
\begin{equation}
\varepsilon_f \dot \eta_2+\eta_2=-\varepsilon_f \dot\alpha_2.
\label{eq:50}
\end{equation}
Hence, under Assumption 3,
\begin{equation}
\|\eta_2(t)\|
\le
e^{-t/\varepsilon_f}\|\eta_2(0)\|
+
\varepsilon_f \bar d_2\big(1-e^{-t/\varepsilon_f}\big).
\label{eq:51}
\end{equation}
Therefore, the filter mismatch remains bounded, and its asymptotic bound can be made arbitrarily small by selecting $\varepsilon_f$ sufficiently small.

Consider now the Lyapunov function candidate
\begin{equation}
V_3
=
\frac{1}{2}z_1^\top z_1
+\frac{1}{2}z_2^\top M(q)z_2
+\frac{1}{2}\tilde z_3^\top \tilde z_3.
\label{eq:52}
\end{equation}
Using \eqref{eq:46} and the closed-loop relation $\dot{\tilde z}_3=-z_2-K_3\tilde z_3$, its time derivative becomes
\begin{equation}
\dot V_3
=
-z_1^\top K_1 z_1
-z_2^\top K_2 z_2
-\tilde z_3^\top K_3 \tilde z_3
+z_2^\top \eta_2.
\label{eq:53}
\end{equation}
Let $k_i=\lambda_{\min}(K_i)$, $i=1,2,3$, and
$\beta=\min\{k_1,k_2/2,k_3\}$.
Applying Young's inequality to the coupling term gives
\begin{equation}
z_2^\top \eta_2
\le
\frac{k_2}{2}\|z_2\|^2+\frac{1}{2k_2}\|\eta_2\|^2.
\end{equation}
Consequently,
\begin{equation}
\dot V_3
\le
-\beta\big(\|z_1\|^2+\|z_2\|^2+\|\tilde z_3\|^2\big)
+\frac{1}{2k_2}\|\eta_2\|^2.
\label{eq:54}
\end{equation}

\begin{theorem}
Consider the reduced-order system \eqref{eq:21a}--\eqref{eq:21b} under Assumptions 1--3 and Property 1. Suppose that the virtual controls are chosen according to \eqref{eq:26} and \eqref{eq:36}, the filtered virtual control is generated by \eqref{eq:44}, and the command-rate input is selected as in \eqref{eq:48}. Then, for closed-loop trajectories evolving in the compact operating region $\Omega$, the filtered error system $(z_1,z_2,\tilde z_3)$
is uniformly ultimately bounded. Moreover, after the filter transient decays,
\begin{equation}
\limsup_{t\to\infty}
\sqrt{
\|z_1(t)\|^2+\|z_2(t)\|^2+\|\tilde z_3(t)\|^2
}
\le
\sqrt{\frac{c_2}{c_1}}\,
\frac{\varepsilon_f \bar d_2}{\sqrt{2\beta k_2}},
\label{eq:55}
\end{equation}
where $c_1$ and $c_2$ are defined in \eqref{eq:43c}.
\end{theorem}

\begin{proof}
Let $\tilde{\mathbf z}=[z_1^\top,z_2^\top,\tilde z_3^\top]^\top$.
By Assumption~1 and \eqref{eq:52},
$c_1\|\tilde{\mathbf z}\|^2\le V_3\le c_2\|\tilde{\mathbf z}\|^2$.
From \eqref{eq:51},
$\limsup_{t\to\infty}\|\eta_2(t)\|\le\varepsilon_f\bar d_2$.
Using \eqref{eq:54} gives
$\dot V_3\le-(\beta/c_2)V_3+\|\eta_2\|^2/(2k_2)$.
The standard comparison argument therefore yields
$\limsup_{t\to\infty}V_3(t)\le c_2(\varepsilon_f\bar d_2)^2/(2\beta k_2)$.
Combining this with $c_1\|\tilde{\mathbf z}\|^2\le V_3$
establishes uniform ultimate boundedness and the corrected bound in \eqref{eq:55}.
\end{proof}

\begin{remark}
The command filter removes the need for direct differentiation of $\alpha_2$, but it introduces the mismatch term $\eta_2$ into the closed-loop analysis. As a result, the exact exponential convergence obtained for the ideal design is replaced here by a practical tracking result. The ultimate bound is proportional to $\varepsilon_f \bar d_2$ and can therefore be reduced by choosing a sufficiently small filter constant.
\end{remark}

\section{Simulation Results}

The numerical studies address two questions: how accurately the finite-dimensional Galerkin model reproduces the actual delayed plant, and how the nominal and command-filtered controllers of Theorems~1 and~2 perform using this reduced-order representation. A scalar nonlinear mechanical system was used to isolate and clearly visualize the effect of the delayed input channel. The parameters were $J_0=1.40$, $J_1=0.30$, and $\bar g=2.00$, giving
$M(q)=J_0+J_1\cos q$,
$C(q,\dot q)=\frac{1}{2}\frac{dM(q)}{dq}\dot q$, and
$G(q)=\bar g\sin q$.
The reference was $q_d(t)=0.25\sin(0.55t)$.
Simulations were run over $t\in[0,10]$ using \texttt{ode45} with relative and absolute tolerances $10^{-8}$ and $10^{-9}$, respectively. The command-filter constant was $\varepsilon_f=0.05$. For the approximation study,
$\tau\in\{0.01,0.05,0.15,0.30\}$ s,
$N\in\{1,2,4,6,8,10\}$, and
$[q(0),\dot q(0)]^\top=[0.30,-0.10]^\top$.

\begin{figure}[t]
\centering
\includegraphics[width=0.95\linewidth]{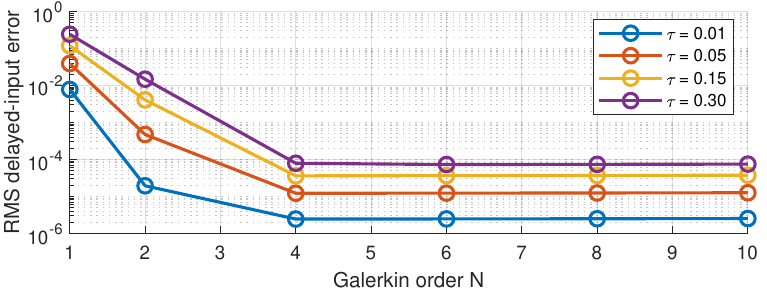}
\caption{Galerkin approximation errors versus Galerkin order.}
\label{fig:galerkin_errors}
\end{figure}

\subsection{Galerkin Approximation Versus Actual Delayed System}

To assess the Galerkin approximation independently of the feedback controller, the smooth command
\[
u(t)=0.12+0.65\sin(1.15t+0.20)+0.18\sin(2.40t-0.35)
\]
was applied to the delay channel. The actual delayed plant used the exact input $u(t-\tau)$, whereas the reduced-order model evolved $a(t)$ through the Galerkin channel dynamics and applied the reconstructed signal $y_N(-\tau,t)=C_Na(t)$.

Fig.~\ref{fig:galerkin_errors} shows the RMS delayed-input error, which directly measures the mismatch between $C_Na$ and $u(t-\tau)$, for increasing $N$ and several delay values. The error decreases rapidly with $N$, with the largest improvement occurring from the lowest orders to moderate-order approximations, indicating that only a few modes capture the dominant delay-channel dynamics. Beyond approximately $N=4$, the curves flatten and additional modes provide only marginal improvement. These results support the reduced-order viewpoint: over the tested delay range, the Galerkin model accurately reproduces the delayed channel while retaining a finite-dimensional realization.

\begin{figure}[t]
    \centering
    \includegraphics[width=\linewidth]{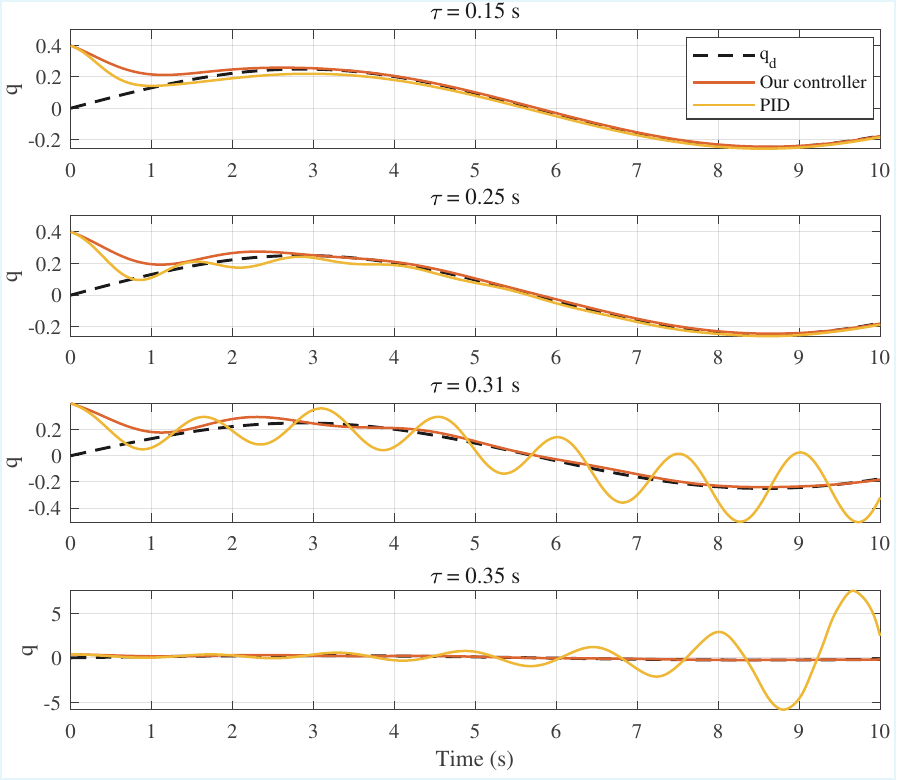}
    \caption{Tracking comparison between the nominal controller in Theorem~1 and a computed-torque PID baseline on the actual delayed plant for increasing input delays.}
    \label{fig:theorem1_vs_pid}
\end{figure}

\subsection{Nominal Controller Versus PID on the Actual Delayed Plant}

The controller of Theorem~1 was next applied to the actual plant with delayed input $u(t-\tau)$ and compared with a computed-torque PID baseline to evaluate the benefit of explicitly modeling the delay channel. As shown in Fig.~\ref{fig:theorem1_vs_pid}, both controllers achieve acceptable tracking for smaller delays, although the proposed controller remains closer to the reference and exhibits a smoother transient. As the delay increases, the PID response develops increasingly pronounced oscillations and departs from the desired trajectory, while the proposed controller remains bounded and well behaved. At the largest delays, the PID exhibits significant oscillatory growth, whereas the proposed controller avoids this severe degradation. Thus, explicitly incorporating the delayed actuator channel into the controller improves delay tolerance over the tested range.

\begin{figure}[t]
    \centering
    \includegraphics[width=\linewidth]{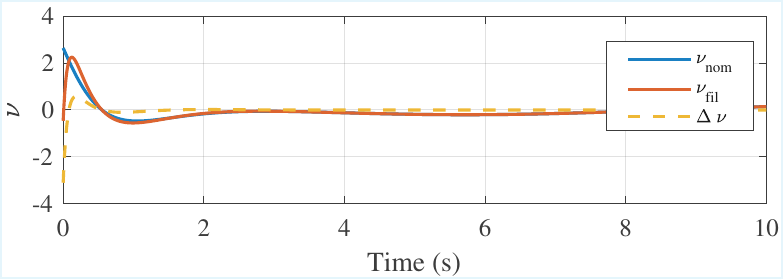}\\[0.5mm]
    {\footnotesize (a) Derivative of control input}\\[0.5mm]
    \includegraphics[width=\linewidth]{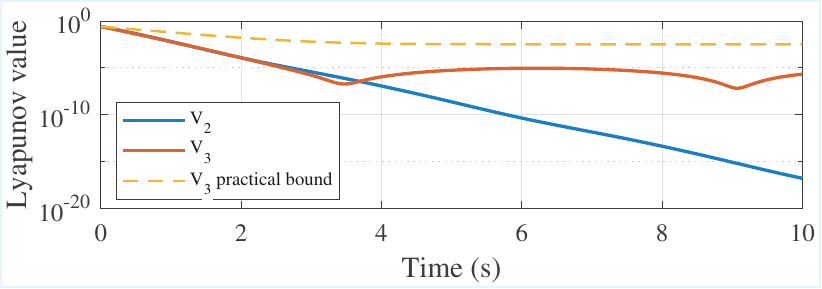}\\[0.5mm]
    {\footnotesize (b) Lyapunov functions and practical bound}
    \caption{Comparison of nominal and command-filtered controllers: (a) control-input derivative and (b) Lyapunov functions with practical bound.}
    \label{fig:nominal_vs_filtered}
\end{figure}

\vspace{-2mm}

\subsection{Nominal Versus Command-Filtered Controller}

Finally, the nominal controller of Theorem~1 is compared with the command-filtered design of Theorem~2 to illustrate the theoretical tradeoff: the nominal design yields exponential convergence for the reduced-order model, whereas the filtered design avoids direct differentiation of the virtual control at the cost of a practical boundedness result.

Fig.~\ref{fig:nominal_vs_filtered}(a) shows that the two command-rate signals become very close after a brief filter transient and their difference rapidly decreases. Thus, the filter preserves the dominant behavior of the nominal backstepping law while avoiding exact differentiation of the virtual control. Fig.~\ref{fig:nominal_vs_filtered}(b) shows that the nominal Lyapunov function decreases toward zero, consistent with Theorem~1, whereas the filtered Lyapunov function initially decreases and then settles to a small nonzero level, consistent with the practical result of Theorem~2. The filtered trajectory also remains below the corresponding practical bound throughout the simulation.

\section{Conclusion}

In this paper, a backstepping-based tracking controller was developed for nonlinear mechanical systems with actuator input delay by approximating the transport-PDE delay channel with a finite-dimensional Galerkin model. Based on this reduced-order model, the nominal design was shown to guarantee uniform exponential stability, while a command-filtered version was developed to avoid direct differentiation and shown to ensure uniform ultimate boundedness. Simulations supported the approach by showing good agreement between the PDE and Galerkin models and improved delay tolerance compared with a PID baseline. Several limitations remain and motivate future work. The analysis is established for the reduced-order Galerkin model rather than the original infinite-dimensional delayed system. Although simulations show that the PDE--Galerkin approximation error is small, its effect on closed-loop stability has not yet been incorporated into the Lyapunov analysis. The current design also assumes availability of the full reduced-order state $(q,\dot q,a)$; observer-based estimation of velocity and channel states will be investigated for settings with noisy or unavailable measurements. In addition, the current study considers a relatively simple system; future work will extend the framework to more complicated systems, investigate unknown delays, and pursue experimental validation.

\bibliographystyle{IEEEtran}
\bibliography{references}

\end{document}